\documentclass[twocolumn,10pt]{autart}  
\usepackage{graphicx}         
\usepackage{amsmath}
\usepackage{enumerate}
\usepackage{graphicx}       
\usepackage{tikz}
\usepackage{apacite}
\usepackage{natbib}
\usepackage{amsfonts}

\usepackage{amssymb}
\usepackage{subfigure}
\usepackage{hyperref}
\hypersetup{
 colorlinks = true,
 citecolor = cyan,
}
\usepackage{wrapfig}
\newenvironment{proof}{{\indent \indent \bf Proof.}}{\hfill $\blacksquare$\par}

\newtheorem{theorem}{Theorem}
\newtheorem{lemma}{Lemma}
\newtheorem{remark}{Remark}

\newtheorem{definition}{Definition}
\newtheorem{assumption}{Assumption}

\newtheorem{example}{Example}
\newtheorem{step}{Step}

\begin{document}

\begin{frontmatter}

\title{Pattern formation using an intrinsic optimal control approach\thanksref{footnoteinfo}} 

\thanks[footnoteinfo]{This work was supported by Natural Science Foundation
of China under Grant T2293772, and the Strategic Priority Research
Program of Chinese Academy of Sciences under Grant No.
XDA27000000. This paper was not presented at any IFAC meeting. Corresponding author: Zhixin Liu.}

\author[amss]{Tianhao Li}\ead{litianhao@amss.ac.cn},    
\author[ntu]{Yibei Li}\ead{yibei.li@ntu.edu.sg}, 
\author[amss]{Zhixin Liu}\ead{lzx@amss.ac.cn},    
\author[kth]{Xiaoming Hu}\ead{hu@kth.se}  

\address[amss]{Key Laboratory of Systems and Control, Academy of Mathematics and Systems Science,
        Chinese Academy of Sciences, and School of Mathematical Sciences, University of Chinese Academy of Sciences, Beijing 100190, P.~R.~China}  

\address[ntu]{School of Electrical and Electronic Engineering, Nanyang Technological University, Singapore 639798, Singapore}

\address[kth]{Department of Mathematics, KTH Royal Institute of Technology, Stockholm 10044, Sweden}             


\begin{abstract}                          
This paper investigates a pattern formation control problem for a multi-agent system modeled with given interaction topology, in which $m$ of the $n$ agents are chosen as leaders and consequently a control signal is added to each of the leaders. These agents interact with each other by Laplacian dynamics on a graph. The pattern formation control problem is formulated as an intrinsic infinite time-horizon linear quadratic optimal control problem, namely, no error information is incorporated in the objective function. Under mild conditions, we show the existence of the optimal control strategy and the convergence to the desired pattern formation. Based on the optimal control strategy, we propose a distributed control strategy to achieve the given pattern. Finally, numerical simulation is given to illustrate theoretical results.
\end{abstract}

\end{frontmatter}

\section{Introduction}
Multi-agent coordination has been an increasingly important modeling method in recent years(cf. \cite{Cao2013,Chen2021}). Formation control problem described in multi-agent coordination which aims to achieve various formation patterns by designing distributed control strategy attracts more and more researchers’ attention (cf. \cite{OH2015}). 

In most existing literatures on the formation control problem, researchers consider the simplest multi-agent systems with agents of single or double integrators \cite{Li2022, Chen2020, Oh2018, de2021, Chen2022, Du2013}. However, in practice, agent systems can rarely be described as simple integrators. Thus, some researchers consider the formation control problem on multi-agent systems with more general and practical agent systems such as general linear systems \cite{Dong2016} and unicycle system \cite{Zhao2018,Kwon2022}. Interactions between agents are essential for multi-agent systems and cause additional difficulty for the control design compared to single agent systems. Nevertheless, as far as we know, in most existing research on the formation control problem, agent interaction only arises in the designed control. Rare existing research consider the formation control problem on multi-agent systems with inherent agent interaction outside the designed control. Naturally, rare literature considers the formation control problem on multi-agent systems when only some agents can be controlled, though this situation is common in reality. 

In this paper, we consider a pattern formation control problem for a multi-agent system with inherent agent interaction in which only some agents are chosen to be leaders, i.e., to be controlled directly. It is noteworthy that the leader-follower framework in this paper is different from that in most existing literature \cite{Tang2021} on the formation control problem, where the leader dynamics defines the reference to which the followers are controlled directly to maintain certain relative formations. In this paper, in contrast, we only design control strategies for the few leader agents while the entire system is driven to the desired pattern with the help of inherent agent interaction among the followers and the leaders. Now, the key problem is how to design the control strategy to drive the whole system state to the desired pattern by only controlling the few leaders. To achieve this, we formulate the pattern formation control problem as
an intrinsic infinite time-horizon linear quadratic optimal control problem, namely, no error information is incorporated in
the objective function. It is well known that the optimal trajectory of regular linear quadratic optimal control problems converges to a fixed point. However, in our problem, the optimal trajectory is expected to converge to a pattern instead of a single point. To overcome this difficulty, we restrict the feasible control set by adding integrators to the system. The optimal control strategy in the reduced feasible control set can drive the system state to the pattern.  By this way, we design a control strategy for the few leaders and  the desired pattern can be achieved with the help of the interaction among the leaders and followers. 

It is clear that to achieve a given pattern, state information of every agent must be used in at least one of the designed controllers and should be known by at least one of the controlled agents. This is natural when all the agents can be controlled. However, when only some agents can be controlled, for agents which are not neighbors of any controlled agents, their state information is not known by any controlled agents following the idea of distributed control. To deal with this difficulty, state information of all the agents are estimated at every controlled agents locally by introducing distributed observers. About distributed observers, see \cite{Yang2022,Han2019,Mitra2018,Park2017,Mitra2017}.  Compared to these well-studied distributed observer methods which aim to obtain estimation of states, we focus more on applying the estimation information obtained by distributed observers to design a distributed control strategy. This idea derives from the classical separation principle in the linear control system. However, for the case of multiagent systems where distributed control strategies are used, classical separation principle fails since the observer at a leader agent can not use the control information of other leader agents except for its neighbors. This results in the observers and controllers must be designed simultaneously for the distributed case. The contribution of this paper is as follows. For Laplacian dynamic systems, under mild conditions, we first design a centralized control strategy to achieve a given pattern and prove the convergence of this control strategy. Then, using method of distributed observer to estimate full states at each controlled agent, we design a distributed control strategy to achieve a given pattern based on the above centralized control strategy. The convergence of distributed control strategy is also proved theoretically. The combination of centralized control strategy and distributed observers indicates a general way for designing distributed control strategy. Finally, simulations are given to illustrate the theoretical results.

The remainder of this paper is organized as follows. In Section \ref{section2}, we first introduce relevant background material, and then introduce the definition of patterns and the formation control problem in this paper. In Section \ref{section3}, optimal control strategies are proposed and analyzed. In Section \ref{section4}, distributed control strategies are proposed and analyzed. In Section \ref{section5}, simulations are given to illustrate the theoretical results. 

\section{Problem formulation}\label{section2}

A graph is denoted as $\mathcal{G}=(\mathcal{V},\mathcal{E})$, where $\mathcal{V}=\{1,2,\cdots,n\}$ is the vertex set and $\mathcal{E}\subset \mathcal{V}\times \mathcal{V}$ is the edge set. The set of neighbors of agent $i$ is denoted by $\mathcal{N}_{i}=\{j\in \mathcal{V}| (j,i)\in \mathcal{E}\}$. Denote the adjacency matrix  of $\mathcal{G}$ as $A$, degree matrix as $D$, and the corresponding  Laplacian matrix as $L$. All the graphs mentioned in the following are undirected graphs. Let $\mathcal{V}_{I}$ be a subset of $\mathcal{V}$. An induced subgraph of $\mathcal{G}$ with respect to the vertex set $\mathcal{V}_{I}$ is denoted as $\mathcal{G}_{I}=(\mathcal{V}_{I},\mathcal{E}_{I})$, where $\mathcal{E}_{I} = \{(j,i)\in \mathcal{V}_{I}\times \mathcal{V}_{I}|(j,i)\in \mathcal{E}\}$.

Consider a multi-agent system on a graph $\mathcal{G}=(\mathcal{V},\mathcal{E})$. The dynamic of the agent $i$ is described by (cf. \cite{Wang2016})
\begin{equation}\label{eq34}
    \dot{x}_{i}(t) = \sum_{j\in \mathcal{N}_{i}}(x_{j}(t)-x_{i}(t)) + a x_{i}(t), i\in \mathcal{V},
\end{equation}
where $x_{i}(t)\in \mathbf{R}$ is the state of the agent $i$, $a$ is a constant and $\mathbf{R}$ denotes the set of real numbers. We see that the self-organization behavior of the system (\ref{eq34}), which is determined by the constant $a$ and the Laplacian matrix $L$ of the graph  $\mathcal{G}$, may not be what we expect. In order to generate desired behaviors, we choose some agents $i_{1},\cdots,i_{m}$ to add control. These chosen agents are called leaders whose dynamics are described by

\begin{equation}\label{eq33}
    \dot{x}_{i}(t) = \sum_{j\in \mathcal{N}_{i}}(x_{j}(t)-x_{i}(t)) + a x_{i}(t) + u_{i}(t),
\end{equation}
where $u_{i}(t)$ is the control added to the leader $i$. The other agents are called followers whose dynamics are still described by (\ref{eq34}). The initial condition of system (\ref{eq34}) and (\ref{eq33}) are denoted as $x_{i}(0)$. 

Denote the set of leader agents as $\mathcal{V}_{l}\triangleq \{i_{1},\cdots,i_{m}\}$ and the set of follower agents as $\mathcal{V}_{f}\triangleq \mathcal{V}-\mathcal{V}_{l}$. The induced subgraph of $\mathcal{G}$ with respect to the vertex set $\mathcal{V}_{l}$ is denoted as $\mathcal{G}_{l}=(\mathcal{V}_{l},\mathcal{E}_{l})$. The set of neighbors of leader node $i_{j}$ in $\mathcal{G}_{l}$ is denoted by $\mathcal{N}^{(l)}_{i_{j}}$. The Laplacian matrix of $\mathcal{G}_{l}$ is denoted as $L_{1}$.



It is clear that the dynamical behavior of followers are influenced by leaders via the graph $\mathcal{G}$. We are interested in how to control the leader systems (\ref{eq33}) such that the state $x(t)=[x_{1}(t),\cdots,x_{n}(t)]^{T}$ of all agents converges to a desired pattern by using only local information. The control strategy is related to form of the desired pattern. Thus, to describe this problem clearly, we next propose a definition about patterns on a graph. 

\begin{definition}\label{definition1}
Given a graph $\mathcal{G}=(\mathcal{V},\mathcal{E})$ and a vector $\alpha = [\alpha_{1},\cdots,\alpha_{n}]$ $(\alpha_{i}=\pm 1,1\leq i\leq n)$, we define a pattern on the graph $\mathcal{G}$ as a set $\mathcal{SP}(\alpha) \triangleq \{x| x=p\alpha,\ p\in \mathbf{R},\  \vert p\vert \geq p_{0} \}$, where $p_{0}$ is a positive constant and the vector $\alpha$ determines the form of patterns.\end{definition}

We note that for practical scenarios, it is hard to measure the pigment concentrations of a pattern (corresponding to $p$ in the definition of $\mathcal{SP}(\alpha)$). Thus, we define the pattern in  Definition \ref{definition1} as a set rather than a point to reflect such a property. 

The following example shows two patterns when the vector $\alpha$ is chosen as two different values.

\begin{example}\label{example1} 
    Take the graph $\mathcal{G}=(\mathcal{V},\mathcal{E})$ as a $7\times 7$ grid graph (see \cite{Notarstefano2013} for grid graphs) and label vertexes of the grid graph from top to bottom and from left to right. Denote $\beta_{1}=[1,-1,1,-1,1,-1,1]^{T}$ and $\beta_{2}=[1,1,1,1,1,1,1]^{T}$. If we choose $\alpha = \hat{\alpha} = \beta_{1} \otimes \beta_{1}$ and the constant $p$ satisfying $\lvert p\rvert \geq p_{0}$, where $\otimes$ denotes the Kronecker product, then Fig.~\ref{fig1_1} shows one element in pattern  $\mathcal{SP}(\hat{\alpha})$, where black squares represent $p$ and white squares represent $-p$. Correspondingly, if the vector $\alpha$  and the constant $p$ are chosen according to  $\bar{\alpha} = \beta_{1} \otimes \beta_{2}$ and  $\lvert p\rvert \geq p_{0}$, then one element in pattern $\mathcal{SP}(\bar{\alpha})$ can be shown in Fig.~\ref{fig1_2}.
\end{example}


\begin{figure}
\begin{center}
\subfigure[Pattern in $\mathcal{SP}_{1}$.]{
    \label{fig1_1}
    \includegraphics[width=0.40\linewidth]{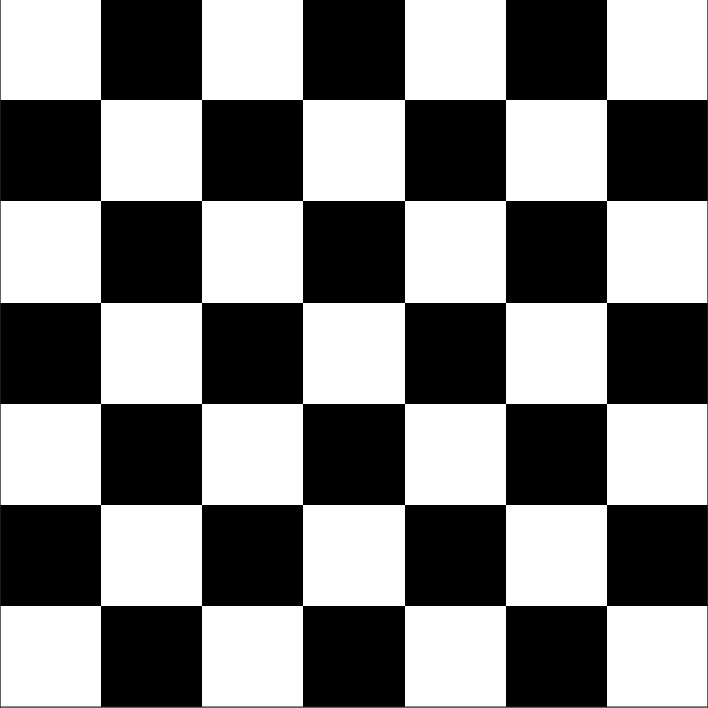}}
\subfigure[Pattern in $\mathcal{SP}_{2}$.]{
    \label{fig1_2}
    \includegraphics[width=0.40\linewidth]{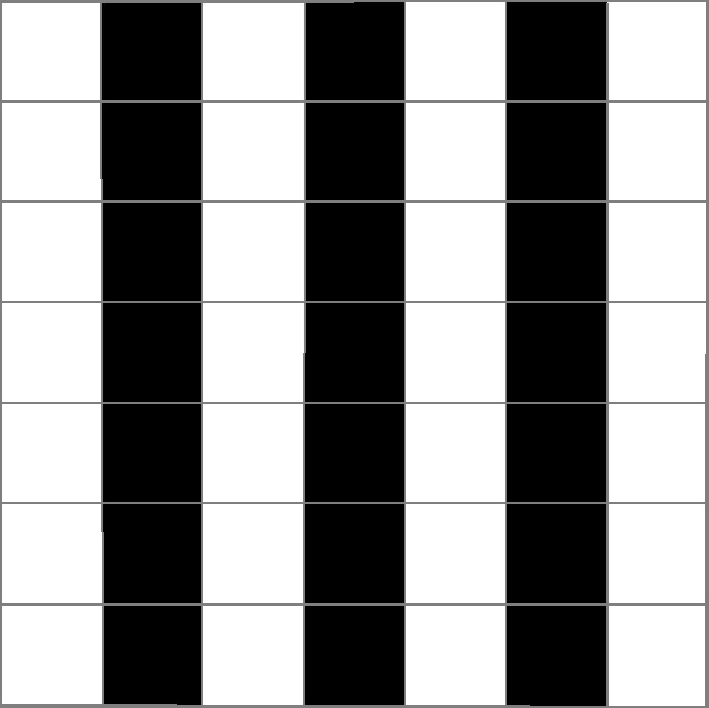}}
\caption{The figures of patterns in $\mathcal{S}_{1}$ and $\mathcal{S}_{2}$.}
\label{fig1}
\end{center}                                 
\end{figure}



In this paper, we focus on designing an optimal control strategy for the leaders so that the state $x(t)$ converges to a given pattern $\mathcal{SP}(\alpha)$. According to Definition \ref{definition1}, we see that the pattern is described by a set $\mathcal{SP}(\alpha)$ rather than a fixed point or a given trajectory, thus the commonly used tracking method (cf., \cite{Feng2024}) may not be applicable to solve this problem. The optimal control method may be an alternative approach to designing the control strategy by designing a performance index related to the desired patterns. In nature, the optimal control strategy would require all state information, which, however, in many practical applications would be difficult to obtain. Therefore, a distributed manner is needed to implement the designed controller.
By using distributed observers, every leader agent can estimate full state of the whole system. Then, replacing the state in centralized control strategies by the estimations obtained from distributed observers the distributed control strategies can be designed. Based on this, we will deal with the problem by the following two steps.

\begin{step}\label{step1}
Design a control strategy so that the state $x(t)$ converges to a given pattern $\mathcal{SP}(\alpha)$ by optimal control method.
\end{step}

\begin{step}\label{step2}
Based on the optimal control strategy in Step \ref{step1}, design a distributed control strategy by using a distributed observer.
\end{step}
In the following two sections, we will design the control strategies and analyze the dynamical behavior of the whole system in the above two steps.

\section{The optimal control strategy}\label{section3}

We rewrite the dynamical systems (\ref{eq34}) and (\ref{eq33}) as the following compact form,
\begin{equation}\label{eq1}
    \dot{x}=-Lx+ax+Bu,
\end{equation}
where $x(t)\in \mathbf{R}^n$ is the state, $u(t)\in \mathbf{R}^{m}$ is the control to be designed, $L$ is the Laplacian matrix of the graph $\mathcal{G}$, $B=[\varepsilon^{(n)}_{i_1},\cdots,\varepsilon^{(n)}_{i_m}]$ is the control matrix, $a$ is a constant and $\varepsilon^{(n)}_{j}$ denotes the $j$th column of the $n\times n$ unit matrix. The initial condition is $x(0)=[x_{1}(0),\cdots,x_{n}(0)]^{T}$.

For Step \ref{step1}, we aim to find a control strategy to drive the state $x(t)$ to the desired pattern $\mathcal{SP}(\alpha)$ by the optimal control method. Normally,  the optimal control problem would be formulated as follows, 
\begin{equation}\nonumber
    \begin{aligned}
        \min_{u} \quad & \bar{J} = \frac{1}{2}\int_{0}^{\infty}( u^{T}Ru + x^{T}Qx)dt,\\
        s.t. \quad &\dot{x}=-Lx+ax+Bu, 
    \end{aligned}
\end{equation}
where positive definite matrix $R$ and positive semi-definite matrix $Q$ are to be designed. However, if we use the linear quadratic optimal control which minimizes the index $\bar{J}$ directly, then the state of (\ref{eq1}) can only be driven to eigenvectors of the Laplacian matrix $L$, where patterns such as stripes we are interested in may not be included. To generate a pattern stably, $\dot{x}(t)$ should converge to zero, and $x(t)$ should converge to a fix point $x_{\infty}$. By the form of the index $\bar{J}$, we know that $u(t)$ must converge to zero. By (\ref{eq1}), it is clear that $x_{\infty}$ is an eigenvector of the the Laplacian matrix $L$. However, taking the stripe pattern defined in Example \ref{example1} and shown in the Fig.~\ref{fig1_2} as an example, by direct calculation, it can be easily verified that the vector $\alpha$ corresponding to the stripe pattern is not an eigenvector of $L$, the Laplacian matrix of the $7\times 7$ grid graph. In a word, some patterns we are interested in such as stripes  can not be generated by controls which converge to zero, such as the control obtained by  directly minimizing the index $\bar{J}$. Thus,
to achieve more desired patterns, we should allow the control $u(t)$ to converge to a nonzero value. In this work, we consider only the case where the value is constant. Consequently, a way to achieve this goal is to add an integrator to the system. 

We propose the following control strategy with integrators to drive the state of the system (\ref{eq1}) to a pattern $\mathcal{SP}(\alpha)$,
\begin{equation}\label{eq2}
    u = z,\ \dot{z} = v,
\end{equation}
where the new control $v$ is designed by minimizing the following performance index,
\begin{equation}\label{eq3}
    J = \frac{1}{2}\int_{0}^{\infty}(\Vert v\Vert^2 + x^{T}Qx)dt,
\end{equation}
and $Q$ is a positive semi-definite matrix to be designed and determines the  convergence point of the system state. Note that in the index (\ref{eq3}), we choose the weight matrix of control as a unit matrix for simplicity.

To design the matrix $Q$ so that the optimal trajectory of the optimal control problem (\ref{eq1})-(\ref{eq3}) converges to the pattern $\mathcal{SP}(\alpha)$, we first transform the pattern $\mathcal{SP}(\alpha)$ to an equivalent form.

Divide the edge set $\mathcal{E}$ in the graph $\mathcal{G}$ into two parts $\mathcal{E}_{1}(\alpha)$ and $\mathcal{E}_{2}(\alpha)$, such that $\mathcal{E}_{1}(\alpha)=\{(i,j)\in \mathcal{E}| \alpha_{i}=-\alpha_{j}\}$ and $\mathcal{E}_{2}(\alpha)=\{(i,j)\in \mathcal{E}| \alpha_{i}=\alpha_{j}\}$. For $k=1,2$, denote $\mathcal{A}^{(k)}(\alpha)=(a^{(k)}_{ij})_{n\times n}$, where $a^{(k)}_{ij}=1$ if and only if $(i,j) \in \mathcal{E}_{k}(\alpha)$ and otherwise $a^{(k)}_{ij}=0$. Denote $\mathcal{N}^{(k)}_{i}(\alpha)=\{j\in \mathcal{V}| (j,i)\in \mathcal{E}_{k}(\alpha)\}$ and $\mathcal{S}(\alpha) = \{x\in \mathbf{R}^{n}| (D+\mathcal{A}^{(1)}(\alpha)-\mathcal{A}^{(2)}(\alpha))x=0,\ \Vert x\Vert \geq p_{0}\sqrt{n} \}$, where $D$ is the degree matrix of the graph $\mathcal{G}$.

To illustrate the relation between $\mathcal{SP}(\alpha)$ and $\mathcal{S}(\alpha)$, an assumption is introduced in the following.

\begin{assumption}\label{assumption1}
    The graph $\mathcal{G}=(\mathcal{V},\mathcal{E})$ is connected.
\end{assumption}

Under the above assumption, we have the following result.

\begin{lemma}\label{lemma1}
    Under Assumption \ref{assumption1}, we have $\mathcal{SP}(\alpha)=\mathcal{S}(\alpha)$. 
\end{lemma} 
\begin{proof}
 According to the definition of $\mathcal{S}(\alpha)$, we see that it can be expressed into the following form equivalently,
\begin{equation}\label{eq31}
\begin{aligned}
    \mathcal{S}(\alpha)=& \{x| x_{i}=-x_{j}, for\  (i,j)\in \mathcal{E}_{1}(\alpha)\ and \\
          & x_{i}=x_{j}, for\  (i,j)\in \mathcal{E}_{2}(\alpha),\ \Vert x\Vert \geq p_{0}\sqrt{n} \}.
\end{aligned}
\end{equation}
By the definition of $\mathcal{E}_{1}(\alpha)$ and $\mathcal{E}_{2}(\alpha)$, it is clear that $\mathcal{SP}(\alpha) \subseteq \mathcal{S}(\alpha)$. We next prove $\mathcal{S}(\alpha) \subseteq \mathcal{SP}(\alpha)$. By Assumption \ref{assumption1}, there exists a path between the node $1$ and any node $j$ in the graph $\mathcal{G}$. Denote nodes on the path as $1,k_{2},\cdots,k_{r-1},j$. By (\ref{eq31}) and the definition of $\mathcal{E}_{1}(\alpha)$ and $\mathcal{E}_{2}(\alpha)$, we have 
\begin{equation}\label{eq32} 
    x_{i}=\frac{\alpha_{i}}{\alpha_{j}}x_{j}
\end{equation}
for any $x = [x_{1},\cdots,x_{n}]^{T}\in \mathcal{S}(\alpha)$ and $(i,j)\in \mathcal{E}$. Thus, we have by (\ref{eq32})
\begin{equation}\nonumber
    x_{j} = \frac{\alpha_{j}}{\alpha_{k_{r-1}}}x_{k_{r-1}}=\cdots=\frac{\alpha_{j}}{\alpha_{1}}x_{1}
\end{equation}
for $1\leq j\leq n$, which indicates $x = p\alpha \in \mathcal{SP}(\alpha)$, where $p=\frac{x_{1}}{\alpha_{1}}$. As $x\in \mathcal{S}(\alpha)$, we have $\lvert p\rvert =\frac{\Vert x \Vert}{\Vert \alpha \Vert}\geq p_{0}$. Thus, $x\in \mathcal{SP}(\alpha)$, which deduces $\mathcal{S}(\alpha) \subseteq \mathcal{SP}(\alpha)$. This completes the proof. 
\end{proof}

Lemma \ref{lemma1} indicates that if the graph $\mathcal{G}$ is connected, then a pattern $\mathcal{SP}(\alpha)$ can also be described by the set $\mathcal{S}(\alpha)$. In the following, we also call the set $\mathcal{S}(\alpha)$ a pattern and our objective is transformed to design a strategy to drive the state of the system to the set $\mathcal{S}(\alpha)$. 

By the form of the set $\mathcal{S}(\alpha)$ which is equivalent to $\mathcal{SP}(\alpha)$, we can choose the positive semi-definite matrix $Q$ in (\ref{eq3}) as 
\begin{equation}
    \label{eqmatrixQ}
    Q = D+\mathcal{A}^{(1)}(\alpha)-\mathcal{A}^{(2)}(\alpha).
\end{equation}


The above optimization problem (\ref{eq1})-(\ref{eq3}) is summarized to the following linear quadratic optimal control problem, 
\begin{equation}\label{eq4}
    \begin{aligned}
        \min_{v} \quad &J = \frac{1}{2}\int_{0}^{\infty}(\Vert v\Vert^2 + x^{T}\bar{Q}x)dt,\\
        s.t. \quad &\dot{\bar{x}} = \bar{A}\bar{x}+\bar{B}v, \\
        & \bar{x}(0)=\begin{bmatrix}
        x(0) \\ z(0)
    \end{bmatrix}
    \end{aligned}
\end{equation}
where 
\begin{equation}\label{eq5}
\begin{aligned}
    &\bar{x} = \begin{bmatrix}
        x \\ z
    \end{bmatrix},
    \bar{A} = \begin{bmatrix}
        -L+aI_{n} & B \\ 0_{m\times n} & 0_{m\times m}
    \end{bmatrix}, \\
    &\bar{B} = \begin{bmatrix}
             0_{m\times m} \\  I_{n}
    \end{bmatrix},
    \bar{Q} = \begin{bmatrix}
        Q & 0_{n\times m} \\ 0_{m\times n} & 0_{m\times m}
    \end{bmatrix},
\end{aligned}
\end{equation}
and $I_{n}$ denotes the $n\times n$ unit matrix.

To solve the optimal problem (\ref{eq4}) and guarantee that the optimal trajectory $x(t)$ in the system (\ref{eq1}) converges to the given pattern $\mathcal{S}(\alpha)$, we need the following assumptions.
\begin{assumption}\label{assumption2}
There exist $u^{*}\in \mathbf{R}^{m}$ and $x^{*}\in \mathcal{S}(\alpha)$, such that $(-L+aI_{n})x^{*}+Bu^{*}=0$.
\end{assumption}

\begin{remark}\label{remark2}
    Assumption \ref{assumption2} is necessary in the following sense: when the control objective is achieved, the state $x(t)$ will converge to a point in the pattern set $\mathcal{S}(\alpha)$ and $\dot{x}(t)$ will converge to zero. Thus, in the case where $u(t)$ converges to a vector as is considered in this work, Assumption \ref{assumption2} is established by (\ref{eq1}).
\end{remark}

\begin{assumption}\label{assumption3}
(L,B,Q) is controllable and  observable, where $Q$ is defined by (\ref{eqmatrixQ}).
\end{assumption}


We note that $(\bar{A},\bar{Q})$ defined by (\ref{eq4}) is not detectable, thus the classical results on linear quadratic optimal control problems may not be applicable to solve the optimization problem (\ref{eq4}). In the following, we introduce a lemma concerning the solution of  (\ref{eq4}) without requiring the detectability of  $(\bar{A},\bar{Q})$.

\begin{lemma}\label{lemma2}(\cite{Trentelman2001})
Suppose $(\bar{A},\bar{B})$ is controllable and $\bar{Q}$ is positive semi-definite, then there exists a smallest real symmetric positive semi-definite solution $P^{-}$ of the following algebric Riccati equation
\begin{equation}\label{eq7}
    \bar{A}^{T}P+PA-P\bar{B}\bar{B}^{T}P+\bar{Q} = 0,
\end{equation}
that is, for any real symmetric positive semi-definite solution $P$ of the algebraic  Riccati equation (\ref{eq7}), $P-P^{-}$ is positive semi-definite. Furthermore, 
\begin{equation}\label{eq42}
    v(t) = -\bar{B}^{T}P^{-}\bar{x}(t)
\end{equation}
is an optimal solution of the linear quadratic optimal control problem (\ref{eq4}).
\end{lemma}




It is clear that under Assumption \ref{assumption3}, we have $(\bar{A},\bar{B})$ is controllable. Thus, by Lemma \ref{lemma2}, the algebraic Riccati equation (\ref{eq7}) related to the optimization problem (\ref{eq4})  has at least one positive semi-definite solution. In the following, we will consider the control strategy 
\begin{equation}\label{eq6}
    v(t) = -\bar{B}^{T}P\bar{x}(t),
\end{equation}
where $P$ is a real symmetric positive semi-definite solution of  (\ref{eq7}). Particularly, when $P=P^{-}$, the control strategy (\ref{eq6}) is an optimal solution of the linear quadratic optimal control problem (\ref{eq4}). The closed-loop system under the control strategy (\ref{eq6}) is written as follows, 
\begin{equation}\label{eq19}
    \dot{\bar{x}}=\widetilde{A}\bar{x},
\end{equation}
where
\begin{equation}\label{eq27}
    \widetilde{A}=\bar{A}-\bar{B}\bar{B}^{T}P.
\end{equation}

To analyze the dynamical behavior of (\ref{eq19}), we introduce the following lemmas about eigenvalues and eigenvectors of the closed-loop system matrix $\widetilde{A}$.

\begin{lemma}\label{lemma4}
   Suppose that $\widetilde{A}$ has at least one zero eigenvalue, and $h$ is any eigenvector of $\widetilde{A}$ corresponding to zero eigenvalue. Then under assumptions in Lemma \ref{lemma2}, we have $Ph=0$, where $P$ is any positive semi-definite solution of the algebraic Riccati equation (\ref{eq7}).
\end{lemma}
\begin{proof} We consider the Lyapunov function $V(t) = \bar{x}^{T}(t)P\bar{x}(t)$. By direct calculations, we have
\begin{equation}\nonumber
    \dot{V} = -\bar{x}^{T}(t)P\bar{B}\bar{B}^{T}P\bar{x}(t)-\bar{x}^{T}(t)\bar{Q}\bar{x}(t)\leq 0.
\end{equation}
Furthermore, assuming $\dot{V}=0$, we have
\begin{equation}\label{eq21}
    \bar{B}^{T}P\bar{x}(t)=0, \bar{Q}\bar{x}(t)=0.
\end{equation}
By the equations (\ref{eq7}) and (\ref{eq21}), we have
\begin{equation}\label{eq22}
    \bar{A}^{T}P\bar{x}(t)+P\bar{A}\bar{x}(t)=0.
\end{equation}
By  the definition of $\widetilde{A}$ and taking derivative of $\bar{B}^{T}P\bar{x}(t)=0$, we have
\begin{equation}\nonumber
    0=\bar{B}^{T}P\widetilde{A}\bar{x}(t)=\bar{B}^{T}P\bar{A}\bar{x}(t)=-\bar{B}^{T}\bar{A}^{T}P\bar{x}(t),
\end{equation}
where  (\ref{eq21}) and (\ref{eq22}) are used in the above equation. Repeating the above process,  we have for $1\leq k\leq n-1$,
\begin{equation}\nonumber
    \bar{B}^{T}(\bar{A}^{T})^{k}P\bar{x}(t)=0.
\end{equation}
Thus, 
\begin{equation}\nonumber
    \bar{x}^{T}(t)P[\bar{B},\bar{A}B,\cdots,\bar{A}^{n-1}B]=0.
\end{equation}
By the assumption that $(\bar{A},\bar{B})$ is controllable, we have
\begin{equation}\nonumber
    \bar{x}^{T}(t)P=0.
\end{equation}
By the Lasalle invariance principle, we obtain that  $\lim_{t\rightarrow \infty} P\bar{x}(t)=0$ for any initial value $\bar{x}(0)$. When the initial value is taken as $\bar{x}
(0)=h$, the solution of (\ref{eq19}) is $\bar{x}(t)=h$. Thus, $Ph=0$. This completes the proof of the lemma.   
\end{proof}

\begin{lemma}\label{lemma3}
    Under Assumptions \ref{assumption1}-\ref{assumption3}, the closed-loop system matrix $\widetilde{A}$ has only one zero eigenvalue, and all the other eigenvalues of $\widetilde{A}$ have negative real parts. Furthermore,
    \begin{equation}\label{eq30}
    \psi_1 = \begin{bmatrix}
        x^{*} \\ u^{*}
    \end{bmatrix},
\end{equation}
is a right eigenvector subject to zero eigenvalue of $\widetilde{A}$,  where $x^{*}, u^{*}$ are defined in Assumption \ref{assumption2}.
\end{lemma}
\begin{proof} Let $\lambda$ be an eigenvalue of $\widetilde{A}$ whose real part is non-negative, and one of corresponding unit eigenvectors is $h$. Then by the equations (\ref{eq7}) and (\ref{eq27}), we have
\begin{equation}
    \label{eq43}
    \bar{h}^{T}(\widetilde{A}^{T}P+P\widetilde{A}+P\bar{B}\bar{B}^{T}P+\bar{Q})h=0,
\end{equation}
where $\bar{h}$ is the conjugate complex vector of $h$.
By direct calculations, we have 
\begin{equation}
    \nonumber
    Re(\lambda)\bar{h}^{T}Ph=0,\ \bar{B}^{T}Ph=0,\ \bar{Q}h=0,
\end{equation}
which indicates $\bar{A}h = \widetilde{A}h = \lambda h$.
Denote $h=[h_{1}^{T},h_{2}^{T}]^{T}$, where $h_{1}\in \mathbf{C}^{n}$ and $h_{2}\in \mathbf{C}^{m}$. Then we have
\begin{equation}\label{eq28}
    (-L+aI)h_{1}+Bh_{2}=\lambda h_{1}, \lambda h_{2}=0, Qh_{1}=0.
\end{equation}
We prove $\lambda = 0$ by reduction to absurdity. If $\lambda \neq 0$, then by (\ref{eq28}), we have $h_{2}=0$ and
\begin{equation}\label{eq29}
    (-L+aI)h_{1}=\lambda h_{1},  Qh_{1}=0.
\end{equation}
By (\ref{eq29}) and the Assumption \ref{assumption3}, we have $h_{1} = 0$. Thus $h=[h_{1}^{T},h_{2}^{T}]^{T}=0$, which contradicts to the fact that $h$ is a unit eigenvector of $\widetilde{A}$. Thus, we have $\lambda = 0$, which means that  $\widetilde{A}$ has  neither positive real parts eigenvalues nor imaginary eigenvalues. 

We now consider zero eigenvalues and the corresponding eigenvectors of $\widetilde{A}$. For this, we will first  prove that $\widetilde{A}h=0$ if and only if 
\begin{equation}\label{eq24}
    (-L+aI)h_{1}+Bh_{2}=0, Qh_{1}=0,
\end{equation}
where $h=[h_{1}^{T},h_{2}^{T}]^{T}\in \mathbf{C}^{n+m}$. If $\widetilde{A}h=0$, then $h$ is an eigenvector of $\widetilde{A}$ corresponding to the eigenvalue $\lambda=0$. By the above discussion, we have (\ref{eq28}), which indicates that (\ref{eq24}) holds since $\lambda=0$. Conversely, if (\ref{eq24}) holds, then by the definition of $\bar{A}$ and $\bar{Q}$, we have
\begin{equation}\label{eq44}
    \bar{A}h=0, \bar{Q}h=0.
\end{equation}
By (\ref{eq7}) we have
\begin{equation}\label{eq45}
    \bar{h}^{T}(\bar{A}^{T}P+PA-P\bar{B}\bar{B}^{T}P+\bar{Q})h = 0.
\end{equation}
By (\ref{eq44}) and (\ref{eq45}), we have $\bar{B}^{T}Ph=0$, which indicates $\widetilde{A}h=0$ by (\ref{eq27}) and (\ref{eq44}).

By Assumption \ref{assumption1}, Lemma \ref{lemma1} and the definition of $Q$, we know that the solution space of $Qh_{1}=0$  has dimension $1$. Furthermore, by Assumption \ref{assumption2} and the fact that $B$ has full column rank, we know that the dimension of the solution space of (\ref{eq24}) is also $1$. This indicates that $\widetilde{A}$ has only one linearly independent eigenvector subject to zero eigenvalue, and this eigenvector can be expressed as (\ref{eq30}).

Next, we prove that the multiplicity of zero eigenvalue of $\widetilde{A}$ is one. As $\widetilde{A}$ has only one eigenvector subject to zero eigenvalue, we just need to prove that $\widetilde{A}$ has no augmented eigenvector subject to zero eigenvalue. We prove this statement by reduction to absurdity. Assume that $\widetilde{A}$ has an augmented eigenvector $\hat{h}$ subject to zero eigenvalue, then $\widetilde{A}\hat{h}=h$, where $h$ is the eigenvector subject to zero eigenvalue. By Lemma \ref{lemma4}, we have $P\widetilde{A}\hat{h}=Ph=0$. Thus,
\begin{equation}\nonumber
\begin{aligned}
        &\bar{\hat{h}}^{T}(P\bar{B}\bar{B}^{T}P+\bar{Q})\hat{h} \\
        &=\bar{\hat{h}}^{T}(\widetilde{A}^{T}P+P\widetilde{A}+P\bar{B}\bar{B}^{T}P+\bar{Q})\hat{h}=0,
\end{aligned}
\end{equation}
where $\bar{\hat{h}}$ is the conjugate complex vector of $\hat{h}$. By this equation, we have
\begin{equation}\label{eq35}
    \bar{B}^{T}P\hat{h}=0, \bar{Q}\hat{h}=0.
\end{equation}
As $\widetilde{A}\hat{h}=h$, by (\ref{eq35}) and the definition of $\widetilde{A}$, we have
\begin{equation}\label{eq25}
    \bar{A}\hat{h}=h, \bar{Q}\hat{h}=0.
\end{equation}
 By (\ref{eq25}), we have $h_{2}=0$. As $\widetilde{A}h=0$, we have (\ref{eq24}). By $h_{2}=0$ and (\ref{eq24}), we obtain the following equation, 
\begin{equation}\label{eq38}
    (-L+aI_{n})h_{1}=0, Qh_{1}=0.
\end{equation}
By (\ref{eq38}) and Assumption \ref{assumption3}, we have $h_{1} = 0$. Thus $h=[h_{1}^{T},h_{2}^{T}]^{T}=0$, which contradicts to the definition of $h$. This completes the proof of this lemma. 
\end{proof}

By Lemma \ref{lemma3}, we see that under Assumptions \ref{assumption1}-\ref{assumption3}, $\widetilde{A}$ has only one zero eigenvalue, and $\psi_{1}$ defined in (\ref{eq30}) is an right eigenvector subject to zero eigenvalue. Correspondingly, we denote $\hat{\psi}_{1}$ as an left eigenvector of $\widetilde{A}$ satisfying the following equations,
\begin{equation}\label{eq36}
    \hat{\psi}_{1}^{T}\widetilde{A}=0, \hat{\psi}_{1}^{T}\psi_1 = 1.
\end{equation}   

Introduce the following set 
\begin{equation}\nonumber
    U_{1} = \left\{\zeta\Big | \lvert \hat{\psi}_{1}^{T}\zeta \rvert  > \frac{p_{0}\sqrt{n}}{\Vert x^{*} \Vert} \right\}.
\end{equation}
We now show that the desired patterns can be obtained when the initial states $\bar{x}(0)$ of the system (\ref{eq19}) are taken from the set $U_{1}$. 

\begin{theorem}\label{theorem1}
    Under Assumptions \ref{assumption1}-\ref{assumption3}, the states of the closed-loop system (\ref{eq19}) will converge to a limit denoted as $\bar{x}_{\infty}=[x^{T}_{\infty}, z^{T}_{\infty}]^{T}$. If $\bar{x}(0)\in U_{1}$, then $x_{\infty} \in \mathcal{S}(\alpha)$.
\end{theorem}
\begin{proof} Denote all eigenvalues of $\widetilde{A}$ as  $\lambda_i (1\leq i\leq n)$, and the eigenvectors corresponding to   $\lambda_i$ as $\psi_i$. By Lemma \ref{lemma3}, we see that one eigenvalue is zero, and all other eigenvalues have negative real parts. We denote $\lambda_1$ as the zero eigenvalue, and $\psi_1$ defined in (\ref{eq30})  is a corresponding right  eigenvector. By straight calculations, we have
\begin{equation}
\begin{aligned}\nonumber
    &\hat{\bar{x}} \triangleq \lim_{t\rightarrow \infty} \bar{x}(t) = \lim_{t\rightarrow \infty}\mathbf{e}^{\widetilde{A}t}\bar{x}(0) \\
    &= T diag(1,0,\cdots,0)T^{-1}\bar{x}(0),
\end{aligned}
\end{equation}
where $T=[\psi_1,\cdots,\psi_n]$. Denote $T^{-1}=[\hat{\psi}_1,\cdots,\hat{\psi}_n]^{T}$. Then we have $\hat{\psi}_1^{T}\widetilde{A}=0$ and $\hat{\psi}_1^{T}\psi_1=1$. By straight calculation, we have $\hat{\bar{x}}=(\hat{\psi}_1^{T} \bar{x}(0)) \psi_1$. When $\bar{x}(0) \in U_{1}$, we have $\hat{x}=(\hat{\psi}_1^{T} \bar{x}(0))x^{*} \in \mathcal{S}(\alpha)$. This completes the proof of the theorem.
\end{proof}

By Theorem \ref{theorem1}, we see that the control strategy (\ref{eq2}) and (\ref{eq6}) can  drive the states of the system (\ref{eq1}) to a given pattern $\mathcal{S}(\alpha)$.

\begin{remark}
It is well-known that for the classical linear quadratic optimal control problem, if $(\bar{A},\bar{B})$ is controllable and $(\bar{A},\bar{Q})$ is detectable, then the states of the closed-loop system under optimal control strategy must converge to the origin. While for the linear quadratic optimal control problem (\ref{eq4}) under consideration, $(\bar{A},\bar{B})$ is controllable but  $(\bar{A},\bar{Q})$ is not detectable because $\bar{A}\psi_{1}=0$, $\bar{Q}\psi_{1}=0$ where the nonzero vector $\psi_{1}$ is defined in (\ref{eq30}). This makes it possible to drive  the states of the closed-loop system under the  optimal control strategy (\ref{eq6})  to the desired patterns.
\end{remark}

\section{Distributed control strategy}\label{section4}

In Section \ref{section3}, we show that under the control strategy (\ref{eq6}), the states of the system (\ref{eq4}) can be driven to the desired pattern. However, the control strategy (\ref{eq6}) depends on the states of all agents, which means that 
it is a kind of centralized control method. In this section, we will investigate how to implement the control strategy (\ref{eq6}) in a distributed way as stated in Step \ref{step2}.



We rewrite $\bar{B}v$ in (\ref{eq4}) as $\sum_{j=1}^{m}\bar{B}_{j}v_{j}$, where 
\begin{equation}
    \label{eq37}
    \bar{B}_{j}=\begin{bmatrix} 0_{n,1} \\ \varepsilon^{(m)}_{j}
    \end{bmatrix}
\end{equation}
is the control matrix, $v_{j}$ is the control input on the leader $i_{j}$, $m$ is the number of leaders, and $\varepsilon^{(m)}_{j}$ is the $j$th column of the $m\times m$ unit matrix. 

For each leader agent $i_{j}$, it can obtain its own system states $x_{i_{j}}$ and integrator states $z_{j}$, that is the vector
\begin{equation}\label{eq8}
    o_j = C_{j}\bar{x},
\end{equation}
where 
    $$C_{j}=\begin{bmatrix} (\varepsilon^{(n)}_{i_{j}})^{T} & 0_{1,m} \\ 0_{1,n} & (\varepsilon^{(m)}_{j})^{T}
    \end{bmatrix}$$
is a $2\times(m+n)$ matrix with $n$ being the number of agents. In addition, each leader agent can communication with its neighbors in the leader graph $\mathcal{G}_{l}$. 


To implement this control strategy in a distributed way, we will estimate the full state of the system at each control node by a distributed observer, and then replace the state $\bar{x}(t)$ in (\ref{eq6}) by the estimation state $\hat{\bar{x}}_j(t)$ at each control node $i_{j}$.  In order to make the distributed observer effectively, we assume that a leader $i_{j}$ can receive the relative information from all of its neighbors in the graph $\mathcal{G}_{l}$, i.e.,  $\{\hat{\bar{x}}_{i_{k}}(t)-\hat{\bar{x}}_{i_{j}}(t)\}_{i_{k}\in \mathcal{N}^{(l)}_{i_{j}}}$.

For each control node $i_{j}\in \mathcal{V}_{l}$, we propose the following distributed control strategy by using only local information,
\begin{equation}\label{eq9}
    v_{j}(t) = -\bar{B}_{j}^{T}P\hat{\bar{x}}_{j}(t),
\end{equation}
where $P$ is defined in (\ref{eq6}), $\hat{\bar{x}}_{j}(t)$ is an estimation of the state $\bar{x}(t)$, and is obtained by the following distributed observer(c.f., \cite{Yang2022,Han2019}),
\begin{equation}\label{eq10}   
   \dot{\hat{\bar{x}}}_j = N_j \hat{\bar{x}}_j+F_j o_j + \chi T^{-1}\sum_{i_{k}\in \mathcal{N}^{(l)}_{i_{j}}}(\hat{\bar{x}}_k - \hat{\bar{x}}_j), 
\end{equation}
where the matrices $N_j$, $F_j$, $T$ and the constant $\chi$ will be determined later,  $o_j$ is defined in (\ref{eq8}), and $\mathcal{N}^{(l)}_{i_{j}}$ is the set of neighbors of leader agent $i_{j}$ in the leader graph $\mathcal{G}_{l}$.

Define the estimation error as $e_{j}(t) = \hat{\bar{x}}_{j}(t)-\bar{x}(t)$. Then, by direct calculations, we have
\begin{equation}\nonumber
\begin{aligned}
    \dot{e}_{j}(t) = &(N_{j}+\bar{B}\bar{B}^{T}P)e_{j}(t)\\
&+\sum_{k=1}^{m}\bar{B}_{k}\bar{B}_{k}^{T}P(e_{k}(t)-e_{j}(t))\\ 
    &+\chi T^{-1}\sum_{i_{k}\in \mathcal{N}^{(l)}_{i_{j}}}(e_{k}(t)-e_{j}(t)) \\
    &+ (N_{j}+\bar{B}\bar{B}^{T}P-\bar{A}+F_{j}C_{j})\bar{x}(t),
\end{aligned}
\end{equation}
with the initial estimation error as 
\begin{equation}\label{eq39}
    e_{j}(0)= \hat{\bar{x}}_{j}(0)-\bar{x}(0).
\end{equation} 

We introduce the following assumption to assure the convergence of the distributed observer.
\begin{assumption}\label{assumption5}
    The leader graph $\mathcal{G}_{l}$ is connected. 
\end{assumption}

We now show how to take the matrices $N_j$, $F_j$, $T$ and the constant $\chi$. 
First, we take 
$N_j$ as follows, 
\begin{equation}\label{eq41}
    N_{j}=\bar{A}-\bar{B}\bar{B}^{T}P-F_{j}C_{j}.
\end{equation}
 By Assumption \ref{assumption3}, we can easily prove that $(\bar{A},C)$ is observable where $C = [C_{1}^{T},\cdots,C_{m}^{T}]^{T}$ with $C_{j}(1\leq j\leq m)$ being defined in (\ref{eq8}). Thus, there exists a matrix $F$ such that $m\bar{A}-FC$ is Hurwitz. The $(n+m)\times 2$ matrix $F_{j}$ is chosen such that 
 \begin{equation}\label{eq46}
     [F_{1},\cdots,F_{m}]=F.
 \end{equation}
Based on this, the matrix  $T$ can be chosen as the positive definite solution of the following Lyapunov equation,
    \begin{equation}\label{eq11}
        (m\bar{A}-FC)^{T}T+T(m\bar{A}-FC)=-I.
    \end{equation}
For $1\leq i\leq m$, we define $\Lambda_{i} = (\bar{A}-F_{i}C_{i})^{T}T+T(\bar{A}-F_{i}C_{i})$. The constant $\chi$ is chosen as a positive constant satisfying
    \begin{equation}\label{eq16}
    \begin{aligned}
        \chi > &\frac{1}{2\lambda_{2}(L_{1})}\lambda_{max}(\Lambda+(\bar{T}K+K^{T}\bar{T})\\
        &+(\widetilde{\Lambda}+\widetilde{T}^{T}K)^{T}W^{-1}(\widetilde{\Lambda}+\widetilde{T}^{T}K)),    
    \end{aligned}
    \end{equation}
    where 
    \begin{equation}\nonumber
        \begin{aligned}
            &\Lambda = diag(\Lambda_{1},\cdots,\Lambda_{m}),\\
            &\widetilde{\Lambda} = [\Lambda_{1},\cdots,\Lambda_{m}],\\
            &\bar{T} = I_{m}\otimes T,\\
           & \widetilde{T} = \mathbf{1}_{m}^{T}\otimes T,\\
            &K = I_{m}\otimes (-\bar{B}\bar{B}^{T}P)\\
        &\ \ \ \ \ \ \ +\mathbf{1}_{m}\otimes [\bar{B}_{1}\bar{B}_{1}^{T}P,\cdots,\bar{B}_{m}\bar{B}_{m}^{T}P],
        \end{aligned}
    \end{equation}
    $\lambda_{2}(L_{1})$ is the second smallest eigenvalue of $L_{1}$, $L_{1}$ is the Laplacian matrix of the leader graph $\mathcal{G}_{l}$, and $\mathbf{1}_{m}$ represents a vector of $m$ dimension in which all the elements are $1$. Assumption \ref{assumption5} assures $\lambda_{2}(L_{1})>0$.

Thus, we have
\begin{equation}\label{eq13}
\begin{aligned}
    \dot{e}_{j}(t) = &(\bar{A}-F_{j}C_{j})e_{j}(t)\\
    &+\sum_{k=1}^{m}\bar{B}_{j}\bar{B}_{j}^{T}P(e_{k}(t)-e_{j}(t))\\
    &+\chi T^{-1}\sum_{i_{k}\in \mathcal{N}^{(l)}_{i_{j}}}(e_{k}(t)-e_{j}(t)).
\end{aligned}
\end{equation}

According to the choice of the matrices $N_j$, $F_j$, $T$ and the constant $\chi$ in (\ref{eq41})-(\ref{eq16}),  we can derive the following  convergence result of the estimation error $e_{j}(t)(1\leq j\leq m)$.

\begin{theorem}\label{theorem2}
     Under Assumptions \ref{assumption3} and \ref{assumption5}, the estimation error for the distributed observer with the control strategy (\ref{eq9}) and (\ref{eq10}) can converge to zero, i,e.,  for $1\leq j\leq m$, we have
    \begin{equation}\nonumber
        \lim_{t\rightarrow \infty}e_{j}(t)=0.    \end{equation}
    
\end{theorem}
\begin{proof}
  Consider the following Lyapunov function,
\begin{equation}\nonumber
    V(t) = \sum_{i=1}^{m}e_{i}^{T}(t)Te_{i}(t).
\end{equation}
Computing the derivative of $V(t)$ along (\ref{eq13}), we have
\begin{equation}\label{eq14}
\begin{aligned}
    \dot{V} &= e^{T}(t)\Lambda e(t)+e^{T}(t)(\bar{T}K+K^{T}\bar{T})e(t)\\
    &-2\chi e^{T}(t)(L_{1}\otimes I_{n})e(t),
\end{aligned}
\end{equation}
where $e = [e_{1},\cdots,e_{m}]^{T}$. Define $U_{c} = \{\mathbf{1}_{m}\otimes \omega|\omega \in \mathbf{R}^{n}\}$. Decompose $e$ according to such a manner $e = e_{c}+e_{r}$, where $e_{c}\in U_{c}$ and $e_{r}\in U_{c}^{\perp}$. Then, substituting this decomposition into  (\ref{eq14}) yields the following inequality, 
\begin{equation}\label{eq15}
\begin{aligned}
    &\dot{V} \leq -\begin{bmatrix}
        \omega \\ e_{r}
    \end{bmatrix}^{T}\\ 
    &\begin{bmatrix}
        W & -(\widetilde{\Lambda}+\widetilde{T}K)\\
        -(\widetilde{\Lambda}+\widetilde{T}K)^{T} & 2\chi \lambda_{2}(L_{1})I-\Phi
    \end{bmatrix}\begin{bmatrix}
        \omega \\ e_{r}
    \end{bmatrix}.
\end{aligned}
\end{equation}
where $\Phi=(\bar{T}K+K^{T}\bar{T})+\Lambda$. 
By (\ref{eq16}), one obtains that
\begin{equation}\label{eq17}
\begin{aligned}
    &2\chi \lambda_{2}(L_{1})I_{m(n+m)}-\Phi \\
    &-(\widetilde{\Lambda}+\widetilde{T}^{T}K)^{T}W^{-1}(\widetilde{\Lambda}+\widetilde{T}^{T}K)
\end{aligned}
\end{equation}
is positive definite. By (\ref{eq17}) and Schur Complement Lemma, we can derive that  $\dot{V}$ in (\ref{eq15}) is  negative definite. This completes the proof of the theorem.
 \end{proof}

Theorem \ref{theorem2} indicates that if the distributed control strategy (\ref{eq9}) and (\ref{eq10}) is applied to the system (\ref{eq4}), and (\ref{eq10}) is viewed as observers, then the estimation errors will converge to zero.

Rewrite (\ref{eq13}) and (\ref{eq39}) as the following equation,
\begin{equation}\label{eq26}
\begin{aligned}
    &\dot{e}(t)=\hat{W}e(t),\\
    &e(0)=[e_{1}^{T}(0),\cdots,e_{m}^{T}(0)]^{T},
\end{aligned}
\end{equation}
where $e(t)=[e_{1}^{T}(t),\cdots,e_{m}^{T}(t)]^{T}$, and $\hat{W}$ is the corresponding system matrix. 

By the above discussion, $\bar{x}(t)$ and $e(t)$ satisfy the following equation,
\begin{equation}\nonumber
    \begin{bmatrix}
        \dot{\bar{x}}(t)\\ \dot{e}(t)
    \end{bmatrix}
    = 
    \hat{M}
    \begin{bmatrix}
        \bar{x}(t)\\ e(t)
    \end{bmatrix},
\end{equation}
where
\begin{equation}\label{eq40}
    \hat{M} = \begin{bmatrix}
        \widetilde{A}& -\hat{K} \\ 0 & \hat{W}
    \end{bmatrix},
\end{equation} 
\begin{equation}\nonumber
    \hat{K}=[\bar{B}_{1}\bar{B}_{1}^{T}P,\cdots,\bar{B}_{m}\bar{B}_{m}^{T}P],
\end{equation}
$\widetilde{A}$ is defined in (\ref{eq27}), and $\bar{B}_{j}$ is defined in (\ref{eq37}). 

Denote 
\begin{equation}\nonumber
\begin{aligned}
    &U_{2} =\\
    &\left\{[\zeta^{T},\eta^{T}]^{T}|\  \lvert \hat{\psi}_{1}^{T}(\zeta+\hat{K}\hat{W}^{-1}\eta) \rvert  > \frac{p_{0}\sqrt{n}}{\Vert x^{*} \Vert} \right\},
\end{aligned}
\end{equation}
where $\hat{\psi}_{1}$ is defined in (\ref{eq36}) and $x^{*}$ is defined in (\ref{eq30}).

For the system (\ref{eq4}) and the distributed control strategy (\ref{eq7}), (\ref{eq9}) and (\ref{eq10}), we have the following results.

\begin{theorem}\label{theorem3}
    Under Assumptions \ref{assumption1}-\ref{assumption5}, we have $\lim_{t\rightarrow \infty} \bar{x}(t)=\lim_{t\rightarrow \infty}\hat{\bar{x}}_{j}(t)$, where $\bar{x}(t)$ and $\hat{\bar{x}}_{j}(t)$ are the states of the system (\ref{eq4}) and (\ref{eq10}). Furthermore, denote the limit of $\bar{x}(t)$ as $\bar{x}_{\infty}=[x^{T}_{\infty},z^{T}_{\infty}]^{T}$,  where $x_{\infty}\in \mathbf{R}^{n}$ and $z_{\infty}\in \mathbf{R}^{m}$. If $[\bar{x}(0)^{T},e^{T}(0)]^{T} \in U_{2}$, then we have $x_{\infty}\in \mathcal{S}(\alpha)$.
\end{theorem}
\begin{proof}
By Theorem \ref{theorem2}, we obtain that the system (\ref{eq26}) is asymptotically stable, which indicates that every eigenvalue of $\hat{W}$ has negative real part.
Denote the eigenvalue set of $\hat{M}$ as $\sigma(\hat{M})$, where $\hat{M}$ is defined in (\ref{eq40}). Then $\sigma(\hat{M})=\sigma(\widetilde{A})\cup \sigma(\hat{W})$. By the results on the eigenvalues of the matrices $\widetilde{A}$ and $\hat{W}$, we obtain that $\hat{M}$ has a zero eigenvalue, and all the other eigenvalues have negative real part. By the structure of $\hat{M}$ and results on eigenvectors of the matrices $\widetilde{A}$, we obtain that the right eigenvector of $\hat{M}$ subject to zero eigenvalue has the form of $[t_1^{T},0_{mn}^{T}]^{T}$, where $t_1$ is the right eigenvector of $\widetilde{A}$ subject to zero eigenvalue, and $0_{mn}$ is a zero vector of $mn$ dimension. The left eigenvector of $\hat{M}$ subject to zero eigenvalue has the form of $[\hat{\psi}_1^{T},\hat{\psi}_1^{T}\hat{K}\hat{W}^{-1}]$, where $\hat{\psi}_1^{T}$ is the left eigenvector of $\widetilde{A}$ subject to zero eigenvalue, and $\hat{K}$ and $\hat{W}$ are defined in (\ref{eq40}). By the way similar to Theorem \ref{theorem1}, we can prove that the limit $\lim_{t\rightarrow \infty} \bar{x}(t)$ exists. By Theorem \ref{theorem2}, we have $\lim_{t\rightarrow \infty} e_{j}(t)=0$, which indicates that $\lim_{t\rightarrow \infty} \hat{\bar{x}}_{j}(t)=\lim_{t\rightarrow \infty} \bar{x}(t)$. Furthermore, if $[\bar{x}^{T}(0),e^{T}(0)]^{T} \in U_{2}$, we can prove $x_{\infty} \in \mathcal{S}(\alpha)$ by a similar way to Theorem \ref{theorem1}. This completes the proof.    
\end{proof}

Theorem \ref{theorem3} indicates that the distributed control strategy (\ref{eq7}), (\ref{eq9}) and (\ref{eq10}) can drive the state of the system (\ref{eq4}) to the given pattern $\mathcal{S}(\alpha)$.

\section{Numerical examples}\label{section5}
In this section, we illustrate the effectiveness of the proposed centralized and distributed control strategies by numerical examples. Take the graph $\mathcal{G}=(\mathcal{V},\mathcal{E})$ as a $3\times 3$ grid graph and number vertexes of the grid graph from top to bottom and from left to right(see Fig.~\ref{fig5_1}). Then, the Laplacian matrix $L$ of graph $\mathcal{G}$ is obtained. Choose 
$$\alpha = \begin{bmatrix}
    1 & 1 & 1 & -1 & -1 & -1 & 1 & 1 & 1
\end{bmatrix}^{T}$$ 
so that the target pattern is shown in Figure \ref{fig5_2}, where blue squares represent  $1$ and  yellow squares represent $-1$.

By the method given in Section \ref{section3}, we can easily obtain the matrix $Q = D+\mathcal{A}^{(1)}(\alpha)-\mathcal{A}^{(2)}(\alpha)$. The control nodes are chosen as $\mathcal{V}_{l}= \{3,2,1,4,7,8,9\}$ and the corresponding control matrix is
\begin{equation}\nonumber
    B = [e^{(9)}_{3},e^{(9)}_{2},e^{(9)}_{1},e^{(9)}_{4},e^{(9)}_{7},e^{(9)}_{8},e^{(9)}_{9}].
\end{equation}
Set $a=4$ and $r=3$. It is clear that Assumptions \ref{assumption1}-\ref{assumption5} are satisfied for the given parameters. The initial condition of the system state and integrator state are chosen arbitrarily between -5 and 5, which are given by
\begin{equation}\nonumber
\begin{aligned}   
   & x(0)=\scriptsize{\begin{bmatrix}
        3.9  &  2.0  &  0.6 &  -3.2  & -2.9 &  -4.2  &  4.1 &   2.1  &  0.6
    \end{bmatrix}^{T}},\\
   & z(0)=\scriptsize{\begin{bmatrix}
        -1.9 &  -3.3 &   1.2  &  4.9  & -3.3 &  -2.4 &  -1.0
    \end{bmatrix}^{T}}.
\end{aligned}
\end{equation}
By calculating left and right eigenvectors of $\widetilde{A}$ subject to zero eigenvalue, we can verify that $\bar{x}(0)=[x(0)^{T},z(0)^{T}]^{T}\in U_{1}$.

\begin{figure}
\begin{center}
\subfigure[The figure of graph.]{
    \label{fig5_1}
    \includegraphics[width=0.40\linewidth]{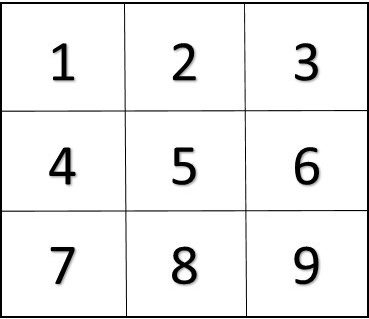}}
\subfigure[The figure of target pattern.]{
    \label{fig5_2}
    \includegraphics[width=0.46\linewidth]{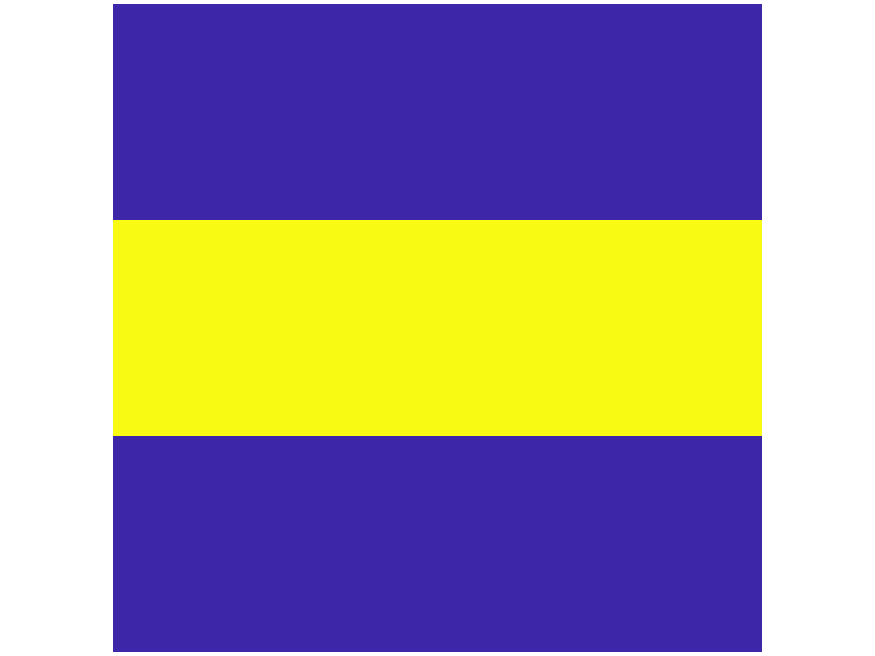}}
\caption{The figures of graph and the target pattern.}
\label{fig5}
\end{center}                                 
\end{figure}


\begin{figure}
\begin{center}
\includegraphics[width=8.0cm]{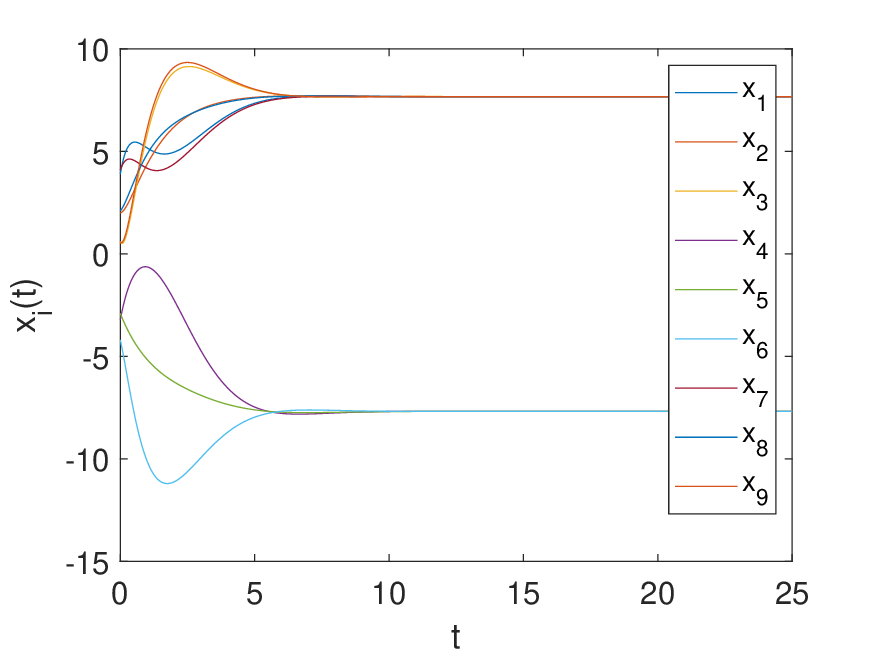} 
  \caption{The state of the system (\ref{eq1}) controlled by centralized control strategy.} 
  \label{fig2}
\end{center}                                 
\end{figure}

\begin{figure}
\begin{center}
\includegraphics[width=8.0cm]{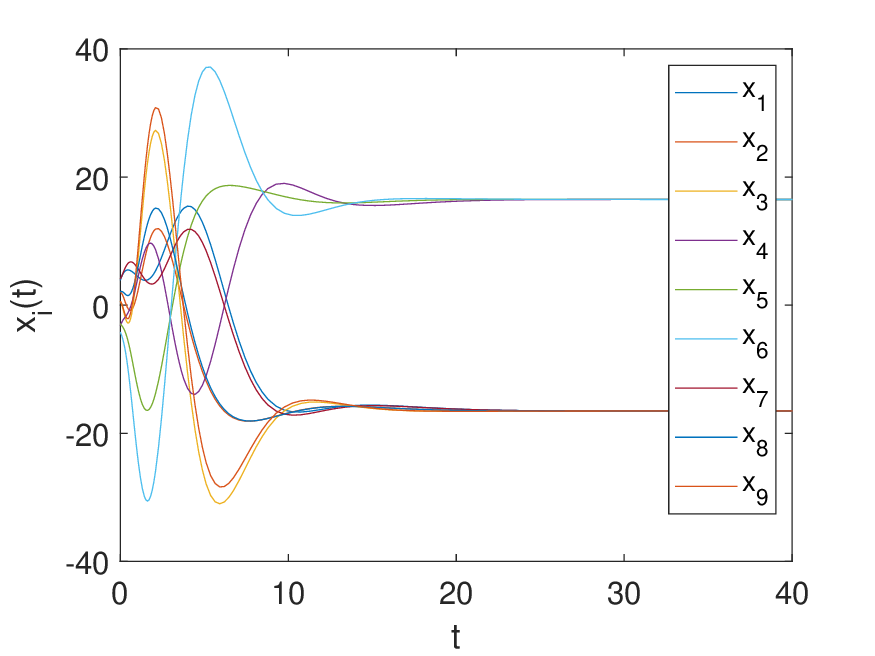} 
  \caption{The state of the system (\ref{eq1}) controlled by distributed control strategy.} 
  \label{fig3}
\end{center}                                 
\end{figure}

In this condition, the state of the system (\ref{eq1}) is shown in Fig.~\ref{fig2} when centralized control strategy (\ref{eq2}) and (\ref{eq6}) are used, in which the matrix $P$ is an arbitrarily chosen positive semi-definite solution of the Riccati equation (\ref{eq7}).
 In Fig.~\ref{fig2}, we see that when the time $t$ is large enough, we have
\begin{equation}\nonumber
    \begin{aligned}
        &x_{1}=x_{2}=x_{3}=-x_{4}=-x_{5}=-x_{6}\\
        &=x_{7}=x_{8}=x_{9}\approx 3.3,
    \end{aligned}
\end{equation}
which indicates that the state $x(t)$ converges to the given pattern $\mathcal{S}(\alpha)$.

For the scenario when the distributed strategy is used, the communication graph $\mathcal{G}_{l}$ of control nodes is a path graph with $7$ nodes, which implies that $\lambda_{2}(L_{1})=0.198$. Distributed state estimation is applied based on distributed observer (\ref{eq10}), where $N_{j}$, $F_{j}$, $T$ are obtained by the following way. $F_{j} = \hat{P}C_{j}^{T}$, where $\hat{P}$ is the only positive definite solution of the following algebraic Ricatti equation,
\begin{equation}\nonumber
    m\bar{A}\hat{P}+m\hat{P}\bar{A}^{T}-\hat{P}C^{T}C\hat{P}+I_{n+m}=0.
\end{equation}
$T$ is given by the solution of (\ref{eq11}), where $W=I_{n+m}$, $F=\hat{P}C^{T}$ and the matrix $C$ is defined in Theorem \ref{theorem2}. $N_{j}$ is given by (\ref{eq41}). Moreover, following (\ref{eq16}), $\chi$ is chosen as $3211$. The initial condition of estimation states $\hat{\bar{x}}_{j,0}(1\leq j\leq m)$ are chosen arbitrarily between $-5$ and $5$. Under these conditions, the state of the system (\ref{eq1}) is shown in Fig.~\ref{fig3} when distributed control strategy (\ref{eq7}), (\ref{eq9}) and (\ref{eq10}) are used. In Fig.~\ref{fig3}, we see that when the time $t$ is large enough, we have
\begin{equation}\nonumber
    \begin{aligned}
        &x_{1}=x_{2}=x_{3}=-x_{4}=-x_{5}=-x_{6}\\
        &=x_{7}=x_{8}=x_{9}\approx 15.0,
    \end{aligned}
\end{equation}
which indicates that the state $x(t)$ converges to the given pattern $\mathcal{S}(\alpha)$.

\section{Concluding remarks}

In this paper, we discuss a formation control problem of a Laplacian dynamic system. This system is a leader-follower multi-agent system with inherent agent interaction and partially controlled agents. To design a distributed control strategy for a formation control problem on systems of this type, a difficulty is that state information of followers which is not neighbor of any leader can not be used in the controller of any leader. To overcome this difficulty, we first design a control strategy to drive the state of the system to a given pattern by optimal control method. Then, when this control strategy is not distributed, by using a distributed observer to estimate the full state of the system at every leader agent, we implement the above control strategy in a distributed way. The combination of centralized control strategies and distributed observers indicates a general way for designing distributed control strategy.

\bibliographystyle{apalike}        
\bibliography{autosam}           

\end{document}